\documentclass[11pt,a4paper]{article}

\usepackage{amsmath}
\usepackage{amsfonts}
\usepackage{amsthm}

\usepackage{graphicx}

\newtheorem{Theorem}{Theorem}
\newtheorem{Corollary}{Corollary}

\newtheorem{Proposition}{Proposition}
\newtheorem{Lemma}{Lemma}
\newtheorem{Remark}{Remark}

\newcommand{\CC}{{\mathrm{C}}}

\newcommand{\ds}{{\mathrm{d}}s}

\newcommand{\dy}{{\mathrm{d}}y}
\newcommand{\Id}{{\mathbf{1}}}
\newcommand{\dom}{{\mathrm{dom}~}}

\newcommand{\var}{\varepsilon}

\newcommand{\partialoperator}{\partial^{Rh}_{s,y}\,}

\begin{document}

\title{A note on the spectrum of the Neumann Laplacian in periodic waveguides}
\author{Carlos R. Mamani {\small and} Alessandra A. Verri \\ 
\vspace{-0.6cm}
\small
\em Departamento de Matem\'{a}tica -- UFSCar, \small \it S\~{a}o Carlos, SP, 13560-970 Brazil\\ \\}


\maketitle

\begin{abstract} 
We study the Neumann Laplacian $-\Delta^N$ restricted to a periodic waveguide.
In this situation its spectrum  $\sigma(-\Delta^N)$ presents a band structure. Our goal and 
strategy is to get spectral information from an analysis of the
asymptotic behavior of these bands provided that the waveguide is sufficiently thin.
\end{abstract}

\

\section{Introduction}

Let $\Lambda$ be a periodic strip (in $\mathbb R^2$) or a periodic tube
(in $\mathbb R^3$). Denote by $-\Delta$  the Laplacian operator  restricted to $\Lambda$.
At the boundary $\partial \Lambda$, 
consider the Dirichlet or  Neumman conditions.
An interesting point is to know something 
about the spectrum $\sigma(-\Delta)$
which has a band structure.

In \cite{yoshitomi} the author  studied the band gap of the spectrum 
of the Dirichlet Laplacian  in a periodic strip in $\mathbb R^2$.
In a more particular situation, in \cite{friedlandersolomyak2} the authors 
studied the band lengths as the diameter of the strip tends to zero.
In \cite{sobolev} the authors 
proved the  absolute continuity for $-\Delta$ in a periodic strip
with either Dirichlet or Neumann conditions.

In the case of periodic tubes,
the absolute continuity was proven in
\cite{bde, friedabscon, vema}.
In \cite{bde, vema} only the Dirichlet boundary condition was considered. 
In \cite{friedabscon} the boundary conditions are more 
general, but a symmetry condition is required.
In \cite{nazarov}, the author
established the existence of gaps in the essential spectrum of the Neumann Laplacian
in a periodic tube.

Consider the Neumann Laplacian $-\Delta^N$ restricted to a periodic waveguide in $\mathbb R^3$. 
This work has two goals. The first one, is to obtain information about the absolutely continuous spectrum of $-\Delta^N$.
The second,
is to prove the existence of band gaps in $\sigma(-\Delta^N)$; although this result
is proven in $\cite{nazarov}$, we give an alternative proof in this text.
We highlight that our purpose is to prove the results above from
an analysis of the
asymptotic behavior of the bands of $\sigma(-\Delta^N)$ provided that the waveguide is sufficiently thin.
Ahead, we give more details.

Let $r: \mathbb R \to \mathbb R^3$ be a simple $C^3$ curve in $\mathbb R^3$
parametrized by its arc-length parameter $s$.
Suppose that $r$ is periodic, i.e., there exists $L > 0$ and a nonzero vector $\vec{u}$ so that
$r(s+L)= \vec{u} + r(s), \forall s \in \mathbb R$. 
Denote by $k(s)$ and $\tau(s)$ the curvature and  torsion of $r$ at the position $s$,
respectively.
Pick $S \neq \emptyset$; an open, bounded, smooth and connected subset of~$\mathbb R^2$. 
Build a waveguide $\Lambda$ in~$\mathbb R^3$ by properly moving the 
region~$S$ along~$r(s)$; at each 
point~$r(s)$ the cross-section region $S$ may present a (continuously differentiable) 
rotation angle $\alpha(s)$.
Suppose that $\alpha(s)$ is  $L$-periodic.
For each  $\var > 0$ (small enough), one can perform this same construction with the region $\var S$ 
and so obtaining a thin waveguide $\Lambda_\varepsilon$.

Now, let $h: \mathbb R \to \mathbb R$ be a $L$-periodic and $C^2$ function satisfying
\begin{equation}\label{conhintro}
0 < c_1 \leq h(s) \leq c_2, \, \forall s \in \mathbb R.
\end{equation}
We consider the thin waveguide, as presented above,
but we deform it by multiplying their cross sections by the function $h(s)$. Thus, we obtain 
a deformed thin tube
$\Omega_\varepsilon$;
see Section \ref{geometrydomain} for details of this construction.

Let $-\Delta_{\Omega_\varepsilon}^N$ be the Neumann Laplacian in $\Omega_\varepsilon$, i.e., the self-adjoint
operator associated with the quadratic form
\begin{equation}\label{quadformneuintro}
\psi \mapsto \int_{\Omega_\varepsilon} |\nabla \psi|^2 d \vec{x}, \quad 
\psi \in  H^1(\Omega_\varepsilon).
\end{equation}

The first result of this work states that

\begin{Theorem}\label{maintheoremintrod}
For each $E > 0$, there exists  $\varepsilon_E > 0$ so that the spectrum of
$-\Delta_{\Omega_\varepsilon}^N$ is absolutely continuous in the interval 
$[0,E]$, for all $\varepsilon \in (0, \varepsilon_E)$.
\end{Theorem}

In \cite{friedabscon} the absolute continuity for $-\Delta_{\Omega_\varepsilon}^N$ was proven under 
the condition of invariance under the reflection
$s \mapsto -s$.

At first, in this introduction, 
we present the main steps of the proof of Theorem \ref{maintheoremintrod};
the details will be presented along the work. Then, we comment our strategy to guarantee
the existence of gaps in the spectrum $\sigma(-\Delta_{\Omega_\varepsilon}^N)$.

Fix a number $c> 0$. Denote by $\Id$ the identity operator. For technical reasons, we 
are going to study the operator $-\Delta_{\Omega_\varepsilon}^N + c \, \Id$;
see Section \ref{reductiondim}.

A change of coordinates shows that $-\Delta_{\Omega_\varepsilon}^N + c \, \Id$
is unitarily equivalent to the operator 
 \begin{equation}\label{operatorassociado}
	T_\varepsilon \psi :=-\frac{1}{h^2 \beta_\varepsilon}\left[\left(\partial_s+{\rm div}_yR^h\right)
	\frac{h^2}{\beta_\varepsilon}\partial^{Rh}_{s,y} \psi+ 
	\frac{1}{\varepsilon^2} {\rm div}_y \left( \beta_\varepsilon \nabla_y \psi \right) \right]  + c \, \psi,
 \end{equation}
 \begin{equation}\label{operatorassociadodom}
\dom T_{\varepsilon}:=	\left\{\psi\in \mathcal{H}^2(\mathbb{R}\times S): \,
\frac{\partial^{Rh}\psi}{\partial N}=0 \quad \text{on} \quad \partial(\mathbb{R}\times S) \right\},
\end{equation}
acting in the Hilbert space $L^2(\mathbb R \times S, h^2 \beta_\varepsilon\ds \dy)$.
Here, $y:=(y_1,y_2) \in S$, ${\rm div_y}$ denotes the divergent of a
vector field in $S$,
\begin{eqnarray}
\beta_\varepsilon(s,y) & :=  & 1 - \varepsilon k(s) (y_1 \cos \alpha(s) + y_2 \sin \alpha(s)),\label{defpeso}\\
 (\partial^{Rh}_{s,y}\psi)(s,y) &:=&\partial_s\psi(s,y)+ \langle\nabla_y\psi(s,y), R^h(s,y)\rangle,\label{opeesq}\\
 R^h(s,y)&:= &  \left(R \, y\right)  (\tau+\alpha')(s) - y \frac{h'(s)}{h(s)}\label{rh},
\end{eqnarray} 
where
$\partial_s \psi:= \partial\psi/\partial s$, 
$\nabla_y\psi:=(\partial\psi/\partial y_1,\partial \psi/\partial y_2)$, 
and  $R$ is the
rotation  matrix $\left( \begin{matrix}
0 & 1\\
-1 & 0
\end{matrix}\right)$.
Furthermore,
\begin{equation}\label{connesq}
\frac{\partial^{Rh}\psi}{\partial N}(s,y) := \frac{h^2(s)}{\beta_\varepsilon(s,y)}
\langle R^h(s,y),N(y)\rangle\partial^{Rh}_{s,y}\psi(s,y)+\frac{\beta_\varepsilon(s,y)}{\varepsilon^2}\langle\nabla_y \psi(s,y), N(y)\rangle;
\end{equation} 
$N$ denotes  the outward point unit normal vector field of $\partial S$.

Since the coefficients of $T_\varepsilon$ are periodic with respect to $s$, we utilize the Floquet-Bloch reduction under the Brillouin zone ${\cal C}:=[-\pi/L, \pi/L]$. More precisely, we show that
$T_\varepsilon$ is unitarily equivalent to the operator
$\int_{{\cal C}}^\oplus T_\varepsilon^\theta \, d \theta,$
where  
\begin{equation}\label{operatorfibraintro}
T^\theta_{\varepsilon} \psi :=-\frac{1}{h^2 \beta_\varepsilon}\left[\left(\partial_s+{\rm div}_y R^h+i\theta\right)
\frac{h^2}{\beta_\varepsilon}\left(\partial^{Rh}_{s,y}+i\theta\right) \psi +
\frac{1}{\varepsilon^2}{\rm div}_y (\beta_\varepsilon \nabla_y \psi) \right]
+ c \, \psi,
\end{equation}
with domain
\begin{eqnarray*}\label{domainoperatorfiber}
& & \dom T^\theta_\varepsilon =   \Big\{ \psi \in \mathcal{H}^2([0,L)\times  S): \\
& & \psi(0,\cdot) =\psi(L,\cdot) \quad \hbox{and} \quad 
\partial^{Rh}_{s,y}\psi(0,\cdot)=\partial^{Rh}_{s,y}\psi(L,\cdot) \quad \hbox{in} \quad L^2(S),\\
& &\, \frac{\partial^{Rh}\psi}{\partial N} =-i\theta  \frac{h^2}{\beta_\varepsilon} \langle R^h,N\rangle\psi \quad
\text{in} \quad L^2([0,L)\times \partial S)\Big\}.
\end{eqnarray*}
Although acting in the Hilbert space $L^2([0,L) \times S, h^2 \beta_\varepsilon \ds \dy)$, 
$\partial_{s,y}^{Rh} \psi$ and 
$\partial^{Rh} \psi / \partial N$ have action given
 by
(\ref{opeesq}), (\ref{rh}) and (\ref{connesq}), respectively.
Furthermore, for each $\theta\in \mathcal{C}$, $T^\theta_\varepsilon$ is self-adjoint;
see Lemma \ref{lemmadecfloq} in Section \ref{floquetbloch} for this decomposition.

Each $T_\varepsilon^\theta$ has compact resolvent and is bounded from below. Thus, 
$\sigma(T_\varepsilon^\theta)$ is discrete.
Denote by $\{E_n(\varepsilon, \theta)\}_{n \in \mathbb N}$ the family of all eigenvalues of $T_\varepsilon^\theta$
and by $\{\psi_n(\varepsilon, \theta)\}_{n \in \mathbb N}$ family of the corresponding normalized  eigenfunctions, i.e.,
\begin{equation}\label{eigenfunctionfloquet}
T_\varepsilon^\theta \psi_n(\varepsilon, \theta) = E_n(\varepsilon, \theta) \psi_n(\varepsilon, \theta),
\quad n =1, 2, 3, \cdots, \quad \theta \in {\cal C}.
\end{equation}

We have
\begin{equation}\label{sigmaband}
\sigma(-\Delta_{\Omega_\varepsilon}^N) = \cup_{n=1}^{\infty} \left\{E_n(\varepsilon, {\cal C})\right\},
\quad \hbox{where} \quad E_n(\varepsilon, {\cal C}) :=  
\cup_{\theta \in {\cal C}} \left\{E_n(\varepsilon, \theta)\right\}.
\end{equation}
Thus, in order to study the spectrum $\sigma(-\Delta_{\Omega_\varepsilon}^N)$, we need 
to analyze each  $E_n(\varepsilon, {\cal C})$ which
is called $n$th band of 
$\sigma(-\Delta_{\Omega_\varepsilon}^N)$.


For each $\theta \in {\cal C}$,
consider the unitary operator ${\cal W}_\theta$ given by (\ref{uniwthetaana})
in Section \ref{analyticity}.
Define
$\tilde{T}_\varepsilon^\theta := {\cal W}_\theta T_\varepsilon^\theta {\cal W}_\theta^{-1}$,
$\dom \tilde{T}_\varepsilon^\theta =
{\cal W}_\theta (\dom T_\varepsilon^\theta)$. 
Due to the definition of ${\cal W}_\theta$,
each domain $\dom \tilde{T}_\varepsilon^\theta$ is independent of $\theta$.
Thus, in that same section,  we prove that 
$\{ \tilde{T}^\theta_\varepsilon, \theta \in {\cal C}\}$
is a type A analytic family. 
This fact ensures that
$E_n(\varepsilon, \theta)$, $n=1,2,3, \cdots,$ are
real analytic functions. 
In addition to this information, another important point to prove Theorem \ref{maintheoremintrod} is 
to know an
asymptotic behavior of the eigenvalues $E_n(\varepsilon, \theta)$ as $\varepsilon$ tends to $0$.
For each $\theta \in {\cal C}$, consider the one dimensional self-adjoint operator
\begin{equation}\label{effectiveoperator}
T^{\theta} w: = (-i \partial_s + \theta)^2 w + 
\frac{h''(s)}{h(s)} w + c \, w, \quad \hbox{in } L^2[0,L), 
\end{equation}
where the functions in
$\dom T^\theta$ satisfy the  conditions $w(0)=w(L)$ and $w'(0)=w'(L).$
For simplicity, write $Q:=[0,L) \times S$.
Define the closed subspace ${\cal L} :=\{w(s) \, 1: w \in L^2[0,L) \} \subset L^2(Q)$.
Note that this subspace is directly related to the fact
that the first eigenvalue of the Neumann Laplacian in a bounded region is zero (and the
constant function is the corresponding eigenfunction).
Consider 
the unitary operators ${\cal X}_\varepsilon$ and $\Pi_\varepsilon$ defined by 
(\ref{uniopereduc}) and (\ref{pilast}), respectively, in Section \ref{reductiondim}.
Our main tool to find an asymptotic behavior for $E_n(\varepsilon, \theta)$
is given by
\begin{Theorem}\label{reductionofdimension}
There exists a number $K > 0$ so that, for all $\varepsilon > 0$ small enough,
\[ \sup_{\theta \in {\cal C}}
\left\{ \left\|{\cal X}_\varepsilon^{-1}\left(T_\varepsilon^\theta \right)^{-1} {\cal X}_\varepsilon 
- \left(\Pi_\varepsilon^{-1}(T^\theta)^{-1} \Pi_\varepsilon \oplus {\bf 0} \right) \right\| \right\} \leq K \, \varepsilon,\]
where ${\bf 0}$ is the null operator on the subspace ${\cal L}^\perp$.
\end{Theorem}

Note that the effective operator $T^\theta$ depends only on a potential induced by 
the deformation $h(s)$.
The bend and twist effects
do not influence
$T_\varepsilon^\theta$. 
This situation change if the Dirichlet condition is considered at the boundary $\partial \Omega_\varepsilon$; see \cite{vema} for a comparison of results.

The spectrum of $T^\theta$ is purely discrete; 
denote by $\nu_n(\theta)$ its $n$th eigenvalue counted with multiplicity.
Let ${\cal K}$ be
a compact subset of ${\cal C}$ which contains an open interval and does not contain the points
$\pm \pi/L$ and $0$.
Given $E>0$, without lost of generality, we can suppose that,
for all $\theta \in {\cal K}$,
the spectrum of $T_\var^\theta$ below $E$ consists of exactly $n_0$ eigenvalues
$\{E_n(\var, \theta)\}_{n=1}^{n_0}$.
As a consequence of Theorem \ref{reductionofdimension},

\begin{Corollary}\label{asymptoticbehaviournew}
For any  $n_0 \in \mathbb N$, there exists $\varepsilon_{n_0} > 0$ so that, 
for all 
$\varepsilon \in (0, \varepsilon_{n_0})$,
\begin{equation}\label{asymptoticintroduction}
E_n(\varepsilon, \theta) =  \nu_n(\theta) + O(\varepsilon),
\end{equation}
holds for each $n=1,2, \cdots, n_0$, uniformly in ${\cal K}$.
\end{Corollary}

\vspace{0.3cm}
\noindent
{\bf Proof of Theorem \ref{maintheoremintrod}:}
Given $E >0$ we can suppose that, for all 
$\theta \in {\cal K}$, the spectrum
of $T_\varepsilon^\theta$ below $E$ consists of exactly $n_0$ eigenvalues 
$\{E_n(\varepsilon, \theta)\}_{n=1}^{n_0}$. 
As already mentioned, the considerations of Section \ref{analyticity} ensure that
$E_n(\varepsilon, \theta)$, $n=1, 2, \cdots, n_0$, are real analityc functions.
The next step is to show that each $E_n(\varepsilon, \theta)$ 
is nonconstant.
Consider the functions $\nu_n(\theta)$, $\theta \in {\cal K}$.
By Theorem XIII.89 in \cite{reedsimon4}, they are  nonconstant.
By Corollary \ref{asymptoticbehaviour}, there exists
$\varepsilon_E > 0$  so that
(\ref{asymptoticintroduction}) holds true for $n=1,2, \cdots, n_0$, 
uniformly in $\theta \in {\cal K}$,
for all $\varepsilon \in (0, \varepsilon_E)$.
Note that $\varepsilon_E > 0$ depends on $n_0$, i.e.,
the thickness of the tube depends on the length
of the energies to be covered.
By Section XIII.16 in \cite{reedsimon4}, the conclusion follows.

\vspace{0.3cm}
As already mentioned, the spectrum of $-\Delta_{\Omega_\varepsilon}^N$  
coincides with the union of  bands; see (\ref{sigmaband}).
It is natural to question  the existence of gaps in its structure.
This subject was studied in \cite{nazarov}. In that work,
the author ensured the existence of gaps. 
However, 
we give an alternative proof for this result.

At first,
it is possible to organize the eigenvalues $\{E_n(\var, \theta)\}_{n \in \mathbb N}$ of  $T_\varepsilon^\theta$ 
in order to obtain a non-decreasing sequence. We keep the same notation and write
\[E_1(\var, \theta) \leq E_2(\var, \theta) \leq \cdots \leq E_n(\var, \theta) \cdots, \quad
\theta \in {\cal C}.\]
In this step 
the functions $E_n(\var, \theta)$ are 
continuous and piece-wise analytic in ${\cal C}$ (see Chapter $7$ in \cite{kato});
each $E_n(\var, {\cal C})$ is either a closed interval or a one point set.
In this case, similar to Corollary \ref{asymptoticbehaviournew}, we have

\begin{Corollary}\label{asymptoticbehaviour}
For any  $n_0 \in \mathbb N$, there exists $\varepsilon_{n_0} > 0$ so that, 
for all 
$\varepsilon \in (0, \varepsilon_{n_0})$,
\begin{equation}\label{asymptoticintroduction}
E_n(\varepsilon, \theta) =  \nu_n(\theta) + O(\varepsilon),
\end{equation}
holds for each $n=1,2, \cdots, n_0$, uniformly in ${\cal C}$.
\end{Corollary}

As a consequence

\begin{Theorem}\label{corgap}
Suppose that $h''(s)/h(s)$ is not constant. Then, there exist $n_1 \in \mathbb N$, $\varepsilon_{n_1+1} > 0$ and $C_{n_1}>0$
so that, for all $\varepsilon \in (0, \varepsilon_{n_1+1})$,
\[\min_{\theta \in {\cal C}} E_{n_1+1} (\varepsilon, \theta) 
-
\max_{\theta \in {\cal C}} E_{n_1} (\varepsilon, \theta) =
C_{n_1} + O(\varepsilon).\]
\end{Theorem}

Theorem \ref{corgap} ensures  that at least one gap appears in the spectrum 
$\sigma(-\Delta_{\Omega_\varepsilon}^N)$, for all $\varepsilon > 0$ small enough.
We highlighted that the deformation at the boundary $\partial \Omega_\varepsilon$ caused by
$h(s)$ generates this effect.
The proof of Theorem \ref{corgap} is based on arguments of
\cite{borg, yoshitomi}.

\begin{Remark}{\rm
Due to the  characteristics of $h$,
if $h$ is not constant, we always have that $h''/h$ is not constant.
In fact, suppose $h''/h=C$.
Without loss of generality, assume $C>0$.
By condition (\ref{conhintro}), we must have $h'' > 0$, i.e., $h'$ is strictly increasing. But this does not occur
because $h'$ is $L$-periodic.

}
\end{Remark}

\begin{Remark}{\rm
Under conditions of Theorems \ref{maintheoremintrod} and \ref{corgap},
we have the existence at least one gap in 
the absolutely continuous spectrum of $-\Delta_{\Omega_\varepsilon}^N$. In fact,
it is enough to choose $\varepsilon > 0$ small enough and an appropriate
$E>0$.
}
\end{Remark}

Although we have proved Theorem \ref{maintheoremintrod} in this Introduction, the proof of Theorem 
\ref{corgap}
will be presented in Section \ref{sectionbandgaps}.


This work is written as follows. 
In Section \ref{geometrydomain}
we construct with details the tube $\Omega_\varepsilon$.
In Section \ref{coordinates} we perform a change of coordinates
so that $\Omega_\varepsilon$ is homeomorphic to the straight tube $\mathbb R \times S$;
as well as the expression for the quadratic form (\ref{quadformneuintro}) in the new variables.
In Section \ref{floquetbloch} we realize the Floquet-Bloch decomposition mentioned in
(\ref{operatorfibraintro}).
In Section \ref{analyticity} we discuss analyticity properties of the functions $E_n(\varepsilon, \theta)$ an
$\psi_n(\varepsilon, \theta)$, $n=1,2,3,\cdots$.
Section \ref{sectioncross} is dedicated to study the Neumann problem in the cross section $S$.
Section \ref{reductiondim} is intended at proofs of Theorem 
\ref{reductionofdimension} and Corollary \ref{asymptoticbehaviour} (the proof of Corollary \ref{asymptoticbehaviournew}
is similar to the proof of Corollary \ref{asymptoticbehaviour}, it will omitted in this text).
In Section \ref{sectionbandgaps} we prove Theorem \ref{corgap}.
A long the text, the symbol $K$ is used to denote different constants and it never depends on 
$\theta$.


\section{Geometry of the domain}\label{geometrydomain}

Let $r : \mathbb R \to \mathbb R^3$ be a simple $C^3$ curve in~$\mathbb R^3$ parametrized by
its arc-length parameter~$s$. We suppose that $r$ is periodic, i.e., there exists $L>0$ 
and a nonzero vector $\vec{u}$
so that
\[r(s+L)=\vec{u}+r(s), \qquad \forall s \in \mathbb R.\]

The curvature of~$r$ at the position~$s$ is $k(s) := \| r''(s)\|$. We choose the usual orthonormal triad of vector 
fields $\{T(s), N(s), B(s)\}$, the so-called Frenet frame, given the tangent, normal
and binormal vectors, respectively, moving along the curve and defined by
\begin{equation}\label{frenet}
T = r'; \quad N = k^{-1} T'; \quad B = T \times N.
\end{equation}
To justify the construction~\eqref{frenet}, it is assumed that $k > 0$, but if~$r$ has a piece of a straight line (i.e., $k = 0$ identically in this piece), usually one can choose a constant Frenet frame instead. It is possible to combine constant Frenet frames with the Frenet frame~\eqref{frenet} 
and so obtaining a global $\CC^2$ Frenet frame; see \cite{klingenberg}, Theorem 1.3.6. In each situation  we assume that a global Frenet frame exists and that the Frenet equations are satisfied, that is,
\begin{equation}\label{frame}
\left(\begin{array}{c}
T' \\
N' \\
B'
\end{array}\right) =
\left( 
\begin{array}{ccc}
0 & k & 0 \\
-k & 0 & \tau \\
0 & -\tau & 0 
\end{array}
\right)
\left(\begin{array}{c}
T \\
N \\
B
\end{array}\right),
\end{equation}
where~$\tau(s)$ is the torsion of $r(s)$, actually defined by~\eqref{frame}. 
Let $\alpha: \mathbb R \to \mathbb R$ be a $L$-periodic and  $C^2$ function
so that $\alpha(0) = 0$, and $S$ an open, bounded,  connected and smooth (nonempty) subset of~$\mathbb R^2$. 
Let $h: \mathbb R \to \mathbb R$ be a $L$-periodic and $C^2$ function satisfying (\ref{conhintro}); see Introduction.
For~$\var > 0$ small enough and $y = (y_1, y_2) \in S$, write
$$\vec{x}(s, y) = r(s) + \var  h(s) y_1 N_\alpha (s) + \var  h(s) y_2 B_\alpha(s)$$
and consider the domain
$$\Omega_\var = \{\vec{x}(s,y) \in \mathbb R^3: s \in \mathbb R, y=(y_1, y_2) \in S\},$$
where
\begin{eqnarray*}
N_\alpha(s) & : = &  \cos \alpha(s) N(s) + \sin \alpha(s) B(s), \\
B_\alpha(s) & := &  - \sin \alpha(s) N(s) + \cos \alpha(s) B(s).
\end{eqnarray*}

Roughly speaking, this tube
$\Omega_\var$ 
is obtained by putting the region $\var h(s) S$ along the curve
$r(s)$, which is simultaneously rotated by an angle $\alpha(s)$ with respect to the
cross section at the position $s= 0$.

\section{Change of coordinates}\label{coordinates}

Consider 
the Neumann Laplacian $-\Delta_{\Omega_\varepsilon}^N$, i.e., the
self-adjoint operator associated with the
quadratic form
\[ 	b_{\varepsilon}(\psi):=\int_{\Omega_{\varepsilon}}|\nabla\psi|^2 {\rm d}\vec{x},\quad \dom b_{\varepsilon}=H^1(\Omega_{\varepsilon}).\]

Fix a number $c> 0$.
For technical reasons, we  consider the quadratic form
 \begin{equation}\label{quadraticformneumann}
 	d^c_{\varepsilon}(\psi):=\int_{\Omega_{\varepsilon}} 
 	\left(|\nabla\psi|^2 + c |\psi|^2 \right) \ds \dy,\quad 
 	\dom d_\varepsilon^c =H^1(\Omega_{\varepsilon}).
 \end{equation}
For simplicity of notation, the symbol $c$ will be omitted;
$d_\varepsilon(\psi):=d_\varepsilon^c(\psi)$.

In this section we perform a change of the variables so that the integration region in 
(\ref{quadraticformneumann}), and consequently the domain of the quadratic form $d_\varepsilon(\psi)$, 
does not depend on $\varepsilon$. For this, consider the mapping 
$$\begin{array}{cccl}
F_{\varepsilon}: & \mathbb{R} \times  S & \rightarrow & \Omega_{\varepsilon}\\
			& (s,y_1,y_2) & \mapsto & r(s)+ \varepsilon h(s) y_1 N_\alpha(s) + \varepsilon h(s) y_2 B_\alpha(s)
\end{array}.$$
Since $h \in  L^{\infty}(\mathbb{R})$, $F_{\varepsilon}$ will be a (global) diffeomorphism for 
$\varepsilon > 0$ small enough. 

In the new variables the domain of $d_\varepsilon(\psi)$ turns to be $H^1(\mathbb{R}\times S)$. On the other hand, the price to be paid is a nontrivial Riemannian metric $G=G^{\alpha,h}_{\varepsilon}$ which is induced by $F_\varepsilon$ i.e., 
$$G=(G_{ij}),\quad G_{ij}=\langle e_i,e_j\rangle, \quad 1\leq i, j\leq 3,$$
where
$$e_1=\frac{\partial F_{\varepsilon}}{\partial s}, \quad e_2=\frac{\partial F_{\varepsilon}}{\partial y_1}, \quad e_3=\frac{\partial F_{\varepsilon}} {\partial y_2}.$$

Some calculations show that in the Frenet frame
\[J:=\left(\begin{matrix}
		e_1\\
		e_2\\
		e_3
	\end{matrix}\right)=\left(\begin{matrix}
	\beta_\varepsilon & \sigma_{\varepsilon} & \delta_{\varepsilon}\\
	0 & \varepsilon h\cos\alpha & \varepsilon h\sin \alpha \\
	0 & -\varepsilon h\sin \alpha & \varepsilon h\cos\alpha
\end{matrix}\right),\]
where  $\beta_\varepsilon(s,y)$ is given by (\ref{defpeso}) in the Introduction, and
\begin{eqnarray*}
\sigma_\varepsilon(s,y) & := & -\varepsilon h(s)(\tau+\alpha')(s)\langle z^{\bot}_{\alpha}(s),y\rangle+ \varepsilon h'(s)\langle z_{\alpha}(s),y\rangle,\\
\delta_\varepsilon(s,y) & := & \varepsilon h(s)(\tau+\alpha')(s)\langle z_{\alpha}(s),y\rangle+\varepsilon h'(s)\langle z^{\bot}_{\alpha}(s),y\rangle,\\
z_{\alpha}(s) & := & (\cos\alpha(s),-\sin\alpha(s)), \\
z^{\bot}_{\alpha}(s) & :=  & (\sin\alpha(s),\cos\alpha(s)).
\end{eqnarray*}

The inverse matrix of $J$ is given by 
$$J^{-1}=\left( \begin{matrix}
\beta_\varepsilon^{-1} & \tilde{\sigma}_{\varepsilon} & \tilde{\delta}_{\varepsilon}\\
0 & (\varepsilon h)^{-1}\cos\alpha & -(\varepsilon h)^{-1}\sin\alpha\\
0 & (\varepsilon h)^{-1}\sin\alpha & (\varepsilon h)^{-1}\cos\alpha 
\end{matrix}\right),$$
where
\[\tilde{\sigma}_\varepsilon(s,y) := \frac{1}{\beta_\varepsilon} \left[(\tau+\alpha')(s) \, y_2-\frac{h'(s)}{h(s)}y_1\right],
\quad 
\tilde{\delta}_{\varepsilon}(s,y) :=  -\frac{1}{\beta_\varepsilon}\left[
(\tau+\alpha')(s) \, y_1-\frac{h'(s)}{h(s)}y_2\right].\]

Note that $JJ^{t}=G$ and $\det J=|\det G|^{1/2}=\varepsilon^2h^2(s) \beta_\varepsilon(s,y)> 0$. 
Thus,
$F_\varepsilon$ is a local diffeomorphism. By requiring that $F_\varepsilon$ is injective (i.e., the tube is not self-intersecting), a global diffeomorphism  is obtained.


Introducing the notation 
\[\|\psi\|^2_G:= \int_{\mathbb{R}\times S}|\psi(s,y)|^2h^2(s)\beta_\varepsilon(s,y) \ds\dy,\]
we obtain a sequence of quadratic forms 
\begin{equation}\label{quadraticfromsemrxs}
	t_\varepsilon(\psi) = \|J^{-1}\nabla\psi\|^2_G + c \| \psi \|_G, \quad 
	\dom t_\varepsilon = H^1(\mathbb R \times S).
	\end{equation}
More precisely, the change of coordinates above is obtained  by the unitary transformation
\[\begin{array}{cccc}
	\Psi_\varepsilon: &  L^2(\Omega_\varepsilon) & \rightarrow & L^2(\mathbb R \times S, 
	h^2 \beta_\varepsilon \ds \dy)\\
	& \psi & \mapsto  & \varepsilon \, \psi \circ F_\varepsilon
	\end{array}.\]
	
After the norms are written out, by (\ref{quadraticfromsemrxs}) we obtain
\[t_\varepsilon(\psi)=
	\int_{\mathbb{R}\times S}
	\left(\frac{h^2}{\beta_\varepsilon}\left|\partial^{Rh}_{s,y}\psi\right|^2+\frac{\beta_\varepsilon}{\varepsilon^2}\left|\nabla_y\psi\right|^2
	+ c \, h^2 \beta_\varepsilon|\psi|^2 \right) \ds \dy,\]
$\dom t_\varepsilon = H^1(\mathbb R \times S)$;
recall the definition of 
$\partial^{Rh}_{s,y}\psi$
in the Introduction.
Note that
$\dom t_\varepsilon$ is a subspace of the Hilbert space $L^2(\mathbb R \times S, h^2 \beta_\varepsilon \ds \dy)$.

Denote by $T_\varepsilon$  the self-adjoint operator associated with the quadratic form $t_\varepsilon(\psi)$. 
In fact, $\Psi_\varepsilon(-\Delta^N_{\Omega_\varepsilon} + c \, \Id)\Psi^{-1}_\varepsilon\psi = T_\varepsilon\psi$, 
$\dom T_\varepsilon = \Psi_\varepsilon(\dom (-\Delta^N_{\Omega_\varepsilon}) )$.
Some calculations show that $T_\varepsilon$ has action and domain given
by (\ref{operatorassociado})  and (\ref{operatorassociadodom}), respectively.
See Appendix A of this work for a discussion about quadratic forms and operators 
associated with them.

\section{Floquet-Bloch decomposition}\label{floquetbloch}

Since the coefficients of $T_\varepsilon$ are periodic with respect to $s$, 
we perform the Floquet -Bloch reduction over the Brillouin zone $\mathcal{C}=[-\pi/L, \pi/L]$. 
For simplicity of notation, we write $\Omega:=\mathbb{R}\times S$ and
\[\mathcal{H}_\varepsilon:=L^2(\Omega,h^2\beta_\varepsilon \ds\dy), 
\quad \mathcal{H}_\varepsilon':= L^2(Q,h^2 \beta_\varepsilon \ds\dy).\]	 
Recall $Q= [0,L) \times S$ and,
for each $\theta \in {\cal C}$, the operator $T_\varepsilon^\theta$ given by 
(\ref{operatorfibraintro}) in the Introduction.

\begin{Lemma}\label{lemmadecfloq}
There exists a unitary operator 
$\mathcal{U}_\varepsilon: \mathcal{H}_\varepsilon \rightarrow \int^\oplus_{\cal C} \mathcal{H}_\varepsilon'\,{\rm d}\theta$, 
so that,
\begin{equation}\label{unitaryfloqlem}
\mathcal{U}_\varepsilon T_\varepsilon \mathcal{U}_\varepsilon^{-1} = \int^{\oplus}_{\mathcal{C}} T^\theta_\varepsilon
\, {\rm d} \theta.
\end{equation}
Furthermore, for each $\theta\in \mathcal{C}$, $T^\theta_\varepsilon$ is self-adjoint.
\end{Lemma}
\begin{proof}
For $(\theta, s,y) \in {\cal C} \times [0,L) \times S$ and $f \in {\cal H}_\varepsilon$
consider the unitary operator
\[{\cal U}_\varepsilon f(\theta,s,y) := \sum_{n \in \mathbb Z} \sqrt{\frac{L}{2\pi}}
e^{-inL\theta-i\theta s} f(s+Ln,y).\]
Some calculations, which will be omitted here,
lead to the formula (\ref{unitaryfloqlem}). 
For the claim that each $T_\varepsilon^\theta$ is self-adjoint, see Appendix A.
\end{proof}

\begin{Remark}{\rm 
For each $\theta \in {\cal C}$, 
the quadratic form $t_\varepsilon^\theta(\psi)$ 
associated with the operator $T_\varepsilon^\theta$ is given by
\[
t_\varepsilon^\theta(\psi) = \int_Q \left(\frac{h^2}{\beta_\varepsilon}
| \partial_{s,y}^{Rh} \psi + i  \, \theta \psi|^2 + \frac{\beta_\varepsilon}{\varepsilon^2} 
|\nabla_y \psi|^2 + c \,h^2 \beta_\varepsilon |\psi|^2 \right) \ds \dy,\]
\[
\dom t_\varepsilon^\theta = \{ \psi \in H^1(Q): \psi(0, \cdot)=\psi(L,\cdot) \, \, 
\hbox{in} \, \,  L^2(S)  \}.
\]
Again, see Appendix A of this work for a discussion about this subject.}
\end{Remark}

\section{Analyticity properties}\label{analyticity}

The goal of this section is to ensure that, for each $n= 1,2, \cdots$, the
functions $E_n(\varepsilon, \theta)$ and $\psi_n(\varepsilon, \theta)$, defined 
by (\ref{eigenfunctionfloquet}) in the Introduction,
are real analytic functions.

The first step is to perform a change of variables in order to turn the domain
$\dom T_\varepsilon^\theta$ independent of the parameter $\theta$. 

Recall the definitions of  $\partial^{R h}/\partial N$ and $R^h$ given by (\ref{connesq}) and (\ref{rh}), respectively;
see Introduction.
Based on \cite{friedabscon}, let $\mu:Q\rightarrow \mathbb R$ be a real function, smooth in the closed 
set $\overline{Q}$, satisfying
\begin{enumerate}
	\item[(1)] $\mu$ is $L$-periodic with respect to $s$, i.e., $\mu(0,y)=\mu(L,y)$, for all $y\in S$;
	\item[(2)] $\displaystyle \frac{\partial^{R h} \mu}{\partial N} = \frac{h^2}{\beta_\varepsilon} \langle R^h, N\rangle$.
\end{enumerate}

Now, define the unitary operator 
\begin{equation}\label{uniwthetaana}
\begin{array}{cccc}
	{\cal W}_\theta: & \mathcal{H}_\varepsilon' & \rightarrow & \mathcal{H}_\varepsilon'\\
	& \eta & \mapsto & e^{i\theta\mu} \, \eta
	\end{array},
	\end{equation}
and 
the self-adjoint operator
\[\label{operator-with-indep-theta-domain}
	\tilde{T}^\theta_\varepsilon = {\cal W}_\theta T^\theta_{\varepsilon} {\cal W}_\theta^{-1}, \quad
	\dom \tilde{T}^\theta_\varepsilon = {\cal W}_\theta( \dom  T^\theta_\varepsilon).\]

Recall the action of $\partial_{s,y}^{Rh} \psi$ by (\ref{opeesq}) (again, see Introduction of this work).
Some straightforward calculations show that 
\begin{eqnarray*}
	\tilde{T}^\theta_{\varepsilon} \psi :
	&=&
	-\frac{1}{h^2 \beta_\varepsilon} \left(\partial_s+{\rm div}_y R^h  
	+ i \theta(\Id -\partialoperator\mu)\right)  \frac{h^2}{\beta_\varepsilon}  \left(\partial^{Rh}_{s,y}
	+	i\theta(\Id -\partialoperator\mu)\right) \psi \\
	& &
	-\frac{1}{\varepsilon^2 h^2 \beta_\varepsilon}\sum_{j=1}^{2}(\partial_{y_j}-i\theta\partial_{y_j}\mu)
	\beta_\varepsilon (\partial_{y_j}-i\theta\partial_{y_j}\mu) \psi + c \, \psi,
\end{eqnarray*}
and,
\begin{eqnarray*}
\dom \tilde{T}^\theta_\varepsilon & = &  
\Big\{ \psi\in\mathcal{H}^2(Q):
\psi(0,\cdot) =\psi(L,\cdot) \quad \hbox{and} \quad
\partialoperator\psi(0,\cdot)=\partialoperator\psi(L,\cdot) \quad
\hbox{in} \quad L^2(S),\\
& &
\, 
\frac{\partial^{Rh}\psi}{\partial N} =0
\quad \text{in} \quad  L^2([0,L)\times \partial S)\Big\}. 
\end{eqnarray*}

Since the domains $\dom \tilde{T}_\varepsilon^\theta$ do not depend on $\theta$, we have 

\begin{Lemma}\label{analiticlem}
$\{ \tilde{T}^\theta_\varepsilon, \theta \in {\cal C}\} $ is a type $A$ analytic family. 
\end{Lemma}

The proof of Lemma \ref{analiticlem} follows the same steps of the proof of Lemma 1 in \cite{vema}. Because
this, it will not be presented here.

Since the operators $T^\theta_\varepsilon$ and  $\tilde{T}^\theta_\varepsilon$
are unitarily equivalent, they have the same spectrum.
Thus, the eigenvalues  of $\tilde{T}^\theta_\varepsilon$ are given by
$E_n(\varepsilon,\theta)$, $n=1,2,3, \cdots$. 
For each $n=1,2,3, \cdots$, the  corresponding eigenfunction is
\[\tilde{\psi}_n(\varepsilon, \theta):=e^{i\theta\mu}\psi_n(\varepsilon,\theta).\]

Lemma \ref{analiticlem} ensures the analyticity of the functions 
$E_n(\varepsilon, \theta)$, $\tilde{\psi}(\varepsilon, \theta)$, $n=1,2,3, \cdots$.
Consequently, the
analyticity of $\psi_n(\varepsilon, \theta)$, $n= 1,2,3, \cdots$.


\section{Cross section problem}\label{sectioncross}

In this section we investigate the Neumann problem in the cross section $S$ which is  an important
step to prove Theorem \ref{reductionofdimension}.

For each $s \in [0,L)$ and $\varepsilon>0$ consider
the Hilbert space 
${\cal H}_\varepsilon^s:=L^2(S, \beta_\varepsilon\dy)$ which is equipped with the inner product 
$\langle u,v \rangle_{{\cal H}_\varepsilon^s}:=\int_S \overline{u} v \beta_\varepsilon \dy.$
Define the quadratic form
\[q_\varepsilon^s(u):=
\int_S |\nabla_y u |^2 \beta_\varepsilon \dy, \quad \dom  q_\varepsilon^s =H^1(S),\]
and denote by $Q_\varepsilon^s$ the self-adjoint operator associated with it.
The geometric features of $S$ ensure that $Q_\varepsilon^s$ has compact resolvent. 
Denote by $\lambda^n_\varepsilon(s)$ the $n$th eigenvalue of $Q^s_\varepsilon$  counted with multiplicity and 
$u^n_\varepsilon(s)$ the corresponding normalized eigenfunction, i.e.,
\[0=\lambda^1_\varepsilon(s)\leq\lambda^2_\varepsilon(s)\leq\lambda^3_\varepsilon(s)\leq\cdots,\]
and
\[Q_\varepsilon^s u^n_\varepsilon(s)= \lambda^n_\varepsilon(s)u^n_\varepsilon(s),\quad n=1,2,3,\cdots.\]

We pay attention that, for each $s \in [0,L)$ and $\varepsilon > 0$, 
$\lambda^1_\varepsilon(s)=0$ and its corresponding eigenfunction $u^1_\varepsilon(s)$ is constant.

Introduce the unitary operator 
\begin{equation*}
\begin{array}{ccccc}
{\cal V}^s_\varepsilon: & L^2(S) &\rightarrow & {\cal H}_\varepsilon^s\\
& u & \mapsto & \beta^{-1/2}_\varepsilon u
\end{array},
\end{equation*}
and define
\[\tilde{q}^s_\varepsilon(u):= q^s_\varepsilon(V^s_\varepsilon u),  \quad 
\dom \tilde{q}_\varepsilon^s:= H^1(S).\]

Some calculations  show that 
\[\tilde{q}^s_\varepsilon(u):=
\int_S\left|\nabla_yu -{\nabla_y \beta_\varepsilon}({2\beta_\varepsilon})^{-1}u\right|^2 \dy,
\quad \dom \tilde{q}^s_\varepsilon:= H^1(S).\]

Let $-\Delta^N_S$ be the Neumann Laplacian operator in $S$, i.e., the self-adjoint operator associated with the quadratic form 
\[q(u):=\int_{S}|\nabla_y u|^2dy, \quad \dom q = H^1(S).\]
Denote by $\lambda^n$ the $n$th eigenvalue of $-\Delta_S^N$
counted with multiplicity  and  by $u_n$ the  corresponding normalized eigenfunction, i.e., 
\[0=\lambda^1<\lambda^2\leq\lambda^3,\cdots,\]
and 
\[-\Delta^N_Su^n=\lambda^nu^n, \quad n=1,2,3,\cdots.\]

\begin{Theorem}\label{neumanncrosssection}
Fix $c_3 > 0$. There exists $K>0$ so that, for all $\varepsilon > 0$ small enough,
\[ \sup_{s \in [0,L)} \left\{ \| ({\cal V}^s_\varepsilon)^{-1} (Q^s_\varepsilon + c_3 \Id)^{-1} {\cal V}^s_\varepsilon -
(-\Delta^N_S + c_3 \Id)^{-1}\| \right\} \leq K \,\varepsilon.\]
\end{Theorem}
\begin{proof}
At first, we add the constant $c_3>0$ only due to a technical detail.
Some calculations show that there exists a number $K>0$ so that, for all $\varepsilon > 0$ small enough,
\[\left|(q_\varepsilon^s(u) +c_3\|u\|_{L^2(S)} ) - (q(u)+c_3\|u\|_{L^2(S)} ) \right|\leq \varepsilon K  \, (q(u)+c_3\|u\|_{L^2(S)}),\]
$\forall u\in H^1(S), \forall s \in [0,L).$
Now, the result follows by Theorem 3 in \cite{bdv}. 
\end{proof}

As a consequence of Theorem \ref{neumanncrosssection}, for all $\varepsilon > 0$ small enough,
\[\left| \frac{1}{\lambda_\varepsilon^2(s) + c_3} - \frac{1}{\lambda^2+c_3} \right|  
\leq \varepsilon \, K, \quad \forall s \in [0,L).\]
Then,
\[0 < \gamma(\varepsilon) \leq \lambda_\varepsilon^2(s), \quad \forall s \in [0,L),\]
where
$\gamma(\varepsilon):=(\lambda^2-\varepsilon c_3 K (\lambda^2+c_3))/(1+\varepsilon K(\lambda^2+c_3)) \to \lambda^2 > 0$,
as $\varepsilon \to 0$.
Thus, there exists $\tilde{\gamma} > 0$ so that, for all $\varepsilon > 0$ small enough,
\begin{equation}\label{crosssectiongamma}
0< \tilde{\gamma} \leq \gamma(\varepsilon) \leq \lambda_\varepsilon^2(s),
\quad \forall s \in [0,L).
\end{equation}


\section{Proof of Theorem \ref{reductionofdimension}
and Corollary \ref{asymptoticbehaviour}}\label{reductiondim}

Recall ${\cal H}_\varepsilon'= L^2(Q, h^2 \beta_\varepsilon \ds\dy)$.
Consider the Hilbert space $ \tilde{\mathcal{H}}_\varepsilon:= {L}^2(Q,\beta_\varepsilon \ds\dy)$ 
equipped with  the inner product 
$\langle \psi, \varphi \rangle_{\tilde{{\cal H}}_\varepsilon} = \int_Q \overline{\psi} \varphi \beta_\varepsilon \ds\dy.$
At first, we perform a change of variables in order to work in $\tilde{\mathcal{H}}_\varepsilon$. 
This change is given by the unitary operator
\begin{equation}\label{uniopereduc}
\begin{array}{cccc}
{\cal X}_\varepsilon: & \tilde{\mathcal{H}}_\varepsilon & \rightarrow & \mathcal{H}'_\varepsilon \\
& \psi & \mapsto &  h^{-1}\psi
\end{array}.
\end{equation}

We start to study  the quadratic form
\[s^{\theta}_\varepsilon(\psi):= t^\theta_\varepsilon({\cal X}_\varepsilon(\psi)),
	\quad \dom s^\theta_\varepsilon:= {\cal X}^{-1}_\varepsilon( \dom t^\theta_\varepsilon).\]

One can show
\begin{eqnarray*}
s^\theta_\varepsilon(\psi) 
&=& 
\int_Q\frac{h^2}{\beta_\varepsilon}\left|\partial^{Rh}_{s,y}(h^{-1}\psi)+i\theta h^{-1} \psi\right|^2 \ds\dy \\ 
& + &
\int_Q\frac{\beta_\varepsilon}{\varepsilon^2}\left|\nabla_y(h^{-1}\psi)\right|^2 \ds\ dy+c\int_Q \left|h^{-1}\psi\right|^2h^2\beta_\varepsilon
\ds \dy \\
&=&
\int_Q\frac{1}{\beta_\varepsilon}\left|\partial^{Rh}_{s,y}\psi+h_\theta(s)\psi\right|^2 \ds \dy \\
& + &
\int_Q\frac{\beta_\varepsilon}{\varepsilon^2 \,h^2} \left|\nabla_y\psi\right|^2 \ds\dy
+ c \int_Q \left|\psi\right|^2\beta_\varepsilon \ds\dy, 
\end{eqnarray*}
where $h_\theta(s):=i \theta- (h'(s)/h(s))$.

Since $h$ is a bounded and $L$-periodic function, 
\[\dom s_\varepsilon^\theta = \{ \psi \in H^1(Q): \psi(0,\cdot) = \psi(L,\cdot) \,\,\, \hbox{in}  
\,\,\, L^2(S)\}.\]
Here, $H^1(Q)$ is a subspace of the Hilbert space $\tilde{\mathcal{H}}_\varepsilon$.

Denote by $S^\theta_\varepsilon$ the self-adjoint operator associated with the quadratic form 
$s^\theta_\varepsilon(\psi)$. Actually,  ${\rm dom}\, S^\theta_\varepsilon \subset {\rm dom}\, s^\theta_\varepsilon$  and 
\[{\cal X}_\varepsilon^{-1}(T^\theta_\varepsilon){\cal X}_\varepsilon = S^\theta_\varepsilon.\]

On the other hand, we define
\begin{eqnarray*}
m^\theta_\varepsilon(\psi)& : =&
\int_Q{\beta_\varepsilon}\left|\partial^{Rh}_{s,y}\psi+h_\theta(s)\psi\right|^2 \ds \dy \\
& + &
\int_Q\frac{\beta_\varepsilon}{\varepsilon^2 \,h^2} \left|\nabla_y\psi\right|^2 \ds\dy
+ c \int_Q \left|\psi\right|^2\beta_\varepsilon \ds\dy, 
\end{eqnarray*}
$\dom m^\theta_\varepsilon:= \dom s^\theta_\varepsilon$.  
Denote by $M^\theta_\varepsilon$ the self-adjoint operator associated with $m^\theta_\varepsilon(\psi)$.
 
\begin{Proposition}\label{proposition1}
There exists a number $K>0$ so that, for all $\varepsilon>0$ small enough, 
\[ \sup_{\theta \in {\cal C}} \left\{
\|(S^\theta_\varepsilon)^{-1}-(M^\theta_\varepsilon)^{-1}\| \right\}\leq K\varepsilon.\]
\end{Proposition}

The main point in this proposition is that $\beta_\varepsilon\rightarrow 1$ uniformly as $\varepsilon\rightarrow 0$. 
Its proof is very similar to the proof of Theorem 3.1 in \cite{dov4} 
and will be omitted here. For technical reasons, we start to study the sequence of operators $M_\varepsilon^\theta$.



Consider the closed subspace $\mathcal{L} = \{w(s)\, 1: w\in {\rm L}^2[0,L)\}$ of the Hilbert space $\tilde{\mathcal{H}}_\varepsilon$.
Take the orthogonal decomposition 
$\tilde{\mathcal{H}}_\varepsilon=\mathcal{L}\oplus\mathcal{L}^{\bot}$.
Thus, for $\psi\in \dom m_\varepsilon^\theta$, one can write 
\begin{equation}\label{decomposition}
\psi(s,y)=w(s)\, 1+ \eta(s,y), \quad w \in H^1[0,L), \eta \in \dom m_\varepsilon^\theta 
\cap {\cal L}^\perp.
\end{equation} 
Furthermore, $w(0)=w(L)$.

Define $ a_\varepsilon(s):=\int_S\beta_\varepsilon(s,y)\dy$ and introduce the Hilbert space ${\cal H}_{a_\varepsilon}:=L^2([0,L),a_\varepsilon\ds)$ equipped whit the inner product
$\langle w_1, w_2 \rangle_{{\cal H}_{a_\varepsilon}} = \int_{0}^{L} \overline{w_1} w_2 a_\varepsilon \ds.$
Acting in  ${\cal H}_{a_\varepsilon}$, consider the one dimensional quadratic form
\begin{eqnarray*}
n^\theta_\varepsilon(w) := m^\theta_\varepsilon(w \, 1) 
& = &  
\int_Q \beta_\varepsilon \left( |(\partial_s +h_\theta)w|^2 + c |w|^2 \right) \ds\dy,\\
& =  &
\int_{0}^{L}\left(a_\varepsilon(s)|(\partial_s +h_\theta) w|^2+ c \, a_\varepsilon(s)|w|^2\right)\ds,
\end{eqnarray*}
$\dom n^{\theta}_\varepsilon:=\{w\in \mathcal{H}^1[0,L); w(0)=w(L) \}$.
Denote by $N^\theta_\varepsilon$ the self-adjoint operator associated with $n^\theta_\varepsilon(w)$.

\vspace{0.3cm}
\noindent
{\bf Proof of Theorem \ref{reductionofdimension}:}
We begin with some observations. 
If $\eta\in \dom m_\varepsilon^\theta \cap \mathcal{L}^\perp$, 
\begin{equation}\label{etaconditionsfirst}
\int_Q w(s) \eta(s,y) \beta_\varepsilon\ds\dy = 0, \quad \forall w \in   \mathcal{L}.
\end{equation}
Consequently,
\begin{equation} \label{etaconditions}
\int_S\eta(s,y)\beta_\varepsilon(s,y) \dy = 0 \quad \hbox{a.e. s},
\end{equation}
{and} 
\begin{equation}\label{etaconditions2}
\int_S \beta_\varepsilon(s,y) \partial_s \eta(s,y) \dy= - \int_S \partial_s\beta_\varepsilon(s,y)\eta(s,y)\dy\quad \hbox{a.e. s}.
\end{equation}
Furthermore, for each $s \in [0,L)$, the Min Max Principle ensures that
\begin{equation}\label{estimativaparasegundoautovalor}
\int_S |\nabla_ y \eta(s,y)|^2 \beta_\varepsilon \dy \geq \lambda_\varepsilon^2(s) \int_S |\eta|^2 \beta_\varepsilon \dy;
\end{equation}
see Section \ref{sectioncross}.

Denote by $m^{\theta}_\varepsilon(\psi_1,\psi_2)$ the sesquilinear form associated with the quadratic form $m^\theta_\varepsilon(\psi)$. 
For $\psi\in \dom m_\varepsilon^\theta$, we consider the decomposition (\ref{decomposition}) and write
\[m^\theta_\varepsilon(\psi)=n^\theta_\varepsilon(w)+ m^\theta_\varepsilon(w \, 1,\eta)+m^\theta_\varepsilon(\eta,w \, 1)+m^\theta_\varepsilon(\eta).\]

We are going to check that there are functions 
$c(\varepsilon)$, $0\leq p(\varepsilon)$ and
$0\leq q(\varepsilon)$, which do not depend on $\theta \in {\cal C}$,
so that $n^\theta_\varepsilon(w)$, $m^\theta_\varepsilon(w \, 1,\eta)$ and  $m^\theta_\varepsilon(\eta)$ satisfy the following conditions:
\begin{equation}\label{c1}
n^\theta_\varepsilon(w) \geq c(\varepsilon)\|w\|_{{\cal H}_{a_\varepsilon}}^{2}, \quad \forall w \in \dom n^\theta_\varepsilon, 
\quad c(\varepsilon)\geq c_0;
\end{equation}
\begin{equation}\label{c2}
	m^\theta_\varepsilon(\eta)\geq p(\varepsilon)\|\eta\|_{\tilde{{\cal H}}_\varepsilon}^2, 
	\quad \forall\eta\in \dom m_\varepsilon^\theta \cap \mathcal{L}^\perp;
\end{equation}
\begin{equation}\label{c3}
|m^\theta_\varepsilon(w \, 1, \eta)|^2\leq q(\varepsilon)^2 n^\theta_\varepsilon (w) \,m^\theta_\varepsilon(\eta), \quad \forall\in \psi \in \dom m^\theta_\varepsilon;
\end{equation}
and with
\begin{equation}\label{c4}
	p(\varepsilon)\rightarrow\infty, \quad c(\varepsilon)=O(p(\varepsilon)), \quad q(\varepsilon)\rightarrow 0 \,\, \text{as}\,\, \varepsilon\rightarrow 0.
\end{equation}
Thus, Proposition 3.1 in \cite{solomyak}, ensures that,
for all $\varepsilon>0$ small enough, 
\begin{equation}\label{resolvent1}
\sup_{\theta \in {\cal C}} \left\{\|(M^\theta_\varepsilon)^{-1}-((N^\theta_\varepsilon)^{-1}\oplus {\bf 0})\|\right\}
\leq  p(\varepsilon)^{-1}
+ K \, q(\varepsilon)c(\varepsilon)^{-1},
\end{equation}
for some number $K>0$.
Recall ${\bf 0}$ is the null operator on the subspace ${\cal L}^\perp$.

Clearly, 
\[n^\theta_\varepsilon(w)\geq c\|w\|_{{\cal H}_{a_\varepsilon}}^2, \quad \forall w\in \dom n_\varepsilon^\theta.\]
By defining  $c(\varepsilon):=c$, it follows the condition (\ref{c1}).

Recall the condition (\ref{conhintro}) in the Introduction.
Note that
\[m^\theta_\varepsilon(\eta)  \geq  \frac{1}{\varepsilon^2}
\int_Q\frac{\beta_\varepsilon}{h^2}|\nabla_y\eta|^2 \ds\dy
\geq 
\frac{1}{\varepsilon^2c_2^2}\int_Q\beta_\varepsilon|\nabla_y\eta|^2\ds\dy,
\quad \forall \eta \in \dom m_\varepsilon^\theta \cap {\cal L}^\perp.\]
By (\ref{crosssectiongamma}) and (\ref{estimativaparasegundoautovalor}), for all $\varepsilon > 0$ small enough,
\[m^\theta_\varepsilon(\eta)  \geq  
\frac{\tilde{\gamma}}{\varepsilon^2c_2^2}\int_Q|\eta|^2 \beta_\varepsilon \ds\dy,
\quad \forall \eta \in \dom m_\varepsilon^\theta \cap {\cal L}^\perp.\]
Just to take 
$p(\varepsilon):=\tilde{\gamma}/ \varepsilon^2 c_2^2$
and then condition (\ref{c2}) is satisfied.

By polarization identity, 
\[m^\theta_\varepsilon(w \, 1,\eta)=
	\int_Q\beta_\varepsilon\overline{\left(\partial^{Rh}_{s,y}+h_\theta\right) w}
	(\partial^{Rh}_{s,y} +h_\theta ) \eta \,\ds\dy 
	+\int_Q\frac{\beta_\varepsilon}{\varepsilon^2 h^2} \langle \nabla_y w,\nabla_y\eta\rangle \ds\dy,\]
which, by (\ref{etaconditionsfirst}) and (\ref{etaconditions}), is simplified to	
\[m^\theta_\varepsilon(w\, 1,\eta)=\int_Q\beta_\varepsilon\overline{(\partial_sw+h_\theta w)}\,\partial_s\eta\,\ds\dy+
\int_Q\beta_\varepsilon\overline{\left(\partial_s w + {h_\theta w} \right)}
\langle\nabla_y \eta, R^h\rangle \,\ds\dy.\]

By (\ref{etaconditions2}),
\[m^\theta_\varepsilon(w1,\eta)=-\int_Q \partial_s(\beta_\varepsilon)\overline{(\partial_sw+h_\theta w)}\eta\,\ds\dy+\int_Q\beta_\varepsilon\overline{(\partial_sw+h_\theta w)}\langle\nabla_y\eta,R^h\rangle\,\ds\dy.\]

Note that there exists $K > 0$ so that 
$|\partial(\beta_\varepsilon)(s,y)|\leq\varepsilon K$,  for all $(s,y) \in Q$.
Since $R^h$ has  bounded coordinates, by H\"older inequality,
\begin{eqnarray*}
|m^\theta_\varepsilon(w \, 1,\eta) |
&\leq & 
K \left( \varepsilon \int_Q|\partial_s w+h_\theta w|\,|\eta|\,\ds\dy + \int_Q\left|\partial_s w + h_\theta w \right| 
\left|\nabla_y \eta \right|\ds \dy\right)\\
&\leq & 
\varepsilon \, K \left(\int_Q|\partial_s w + h_\theta w|^2 \ds\dy\right)^{1/2}\left(\int_Q|\eta|^2\,\ds\dy\right)^{1/2}\\
& +  &
K
\left(\int_Q\beta_\varepsilon|\partial_s w + h_\theta w|^2 \ds\dy\right)^{1/2}\left(\int_Q\beta_\varepsilon|\nabla_y\eta|^2\,\ds\dy\right)^{1/2}\\
&\leq & K\left(n^\theta_\varepsilon(w)\right)^{1/2} \left[\varepsilon \left(m^\theta_\varepsilon(\eta)\right)^{1/2}+ \left(\int_Q\frac{\beta_\varepsilon}{h^2}|\nabla_y\eta|^2\,\ds\dy\right)^{1/2}\right],
\end{eqnarray*}
for all $w \in \dom n_\varepsilon^\theta$, for all $\eta \in \dom m_\varepsilon^\theta \cap {\cal L}^\perp$,
for some $K> 0$, for all $\varepsilon > 0$ small enough.

Now, we can see that
\[|m^\theta_\varepsilon(w \, 1,\eta)| \leq 
K \, \varepsilon \, (n^\theta_\varepsilon(w))^{1/2}(m^\theta_\varepsilon(\eta))^{1/2},
\quad 
\forall w \in \dom n^\theta_\varepsilon, \forall \eta \in \dom m_\varepsilon^\theta \cap {\cal L}^\perp,
\]
for some $K > 0$, for all $\varepsilon > 0$ small enough.

Then, by taking $q(\varepsilon) := K \, \varepsilon$, it is found that conditions (\ref{c3}) and (\ref{c4}) are satisfied. Therefore, 
we finish the proof of (\ref{resolvent1}) where
the upper bound in that inequality is 
$K \, \varepsilon$.

The next step is to study the sequence of one-dimensional operators $N_\varepsilon^\theta$.

In order to work in $L^2[0,L)$ with the usual measure, we define the unitary operator
\begin{equation}\label{pilast}
\begin{array}{cccc}
\Pi_\varepsilon :& L^2[0,L) & \rightarrow & {\cal H}_{a_\varepsilon}\\
& w & \mapsto & a_\varepsilon^{-1/2}w
\end{array},
\end{equation}
and the quadratic form 
\begin{eqnarray*}
o^\theta_\varepsilon(w)&:=& n^\theta_\varepsilon(\Pi_\varepsilon \, w ) \\
&=& \int^L_0\left( |\partial_s w+h_\theta w - (2 \,a_\varepsilon)^{-1}\partial_s(a_\varepsilon) w |^2+c|w|^2\right) \ds,
\end{eqnarray*}
$\dom o^\theta_\varepsilon=\{w\in \mathcal{H}^1[0,L);  w(0)=w(L)\}$. 
Denote by $O^\theta_\varepsilon$ the self-adjoint operator associated with $o^\theta_\varepsilon(w)$.
Note that  $O_\varepsilon^\theta = \Pi_\varepsilon ^{-1} N_\varepsilon^\theta \, \Pi_\varepsilon$.
  
Finally, we define
\[t^\theta(w):=\int^L_0\left(|\partial_s w+h_\theta w|^2+c|w|^2\right) \ds, \quad \dom t^\theta := \dom o^\theta_\varepsilon.\]
The self-adjoint operator associated with
$t^\theta(w)$ is given by $T^\theta$; see (\ref{effectiveoperator}) in the Introduction.

One can show that there exists $K>0$ so that, for all $\varepsilon > 0$ small enough,
\[| o_\varepsilon^\theta(w) - t^\theta(w)| \leq K \, \varepsilon \, t^\theta(w),
\quad \forall w \in \dom t^\theta, \forall \theta \in {\cal C}.\]
Thus, Theorem 3 in \cite{bdv} ensures that, for all $\varepsilon>0$ small enough,
 \begin{equation} \label{resolvent2}
\sup_{\theta \in {\cal C}} \left\{ 
\|(O^\theta_\varepsilon)^{-1}-(T^\theta)^{-1}\| \right\} \leq K\varepsilon.
\end{equation} 
	
It is important to mention that the constants $K$'s, in all this proof,
do not depend on $\theta \in {\cal C}$.

By Proposition \ref{proposition1}, estimates (\ref{resolvent1})   and  (\ref{resolvent2}),  Theorem \ref{reductionofdimension} is proven.

\vspace{0.3cm}
\noindent
{\bf Proof of Corollary  \ref{asymptoticbehaviour}:}
Theorem  \ref{reductionofdimension} in  the Introduction and Corollary 2.3 of 
\cite{gohberg}
imply
\begin{equation}\label{aproxinverse}
\left| \frac{1}{E_n(\varepsilon, \theta)} - \frac{1}{\nu_n(\theta)} \right| \leq K \, \varepsilon,
\quad \forall n \in \mathbb N, \, \forall \theta \in {\cal C},
\end{equation}
for all $\varepsilon > 0$ small enough.
Then,
\[\left|E_n(\varepsilon, \theta) - \nu_n(\theta)\right| \leq 
K \, \varepsilon \, |E_n(\varepsilon, \theta)| \, |\nu_n(\theta)|, \quad \forall n \in \mathbb N, \, \forall \theta \in {\cal C},\]
for all $\varepsilon > 0$ small enough.

The functions $\nu_n(\theta)$ are continuous in ${\cal C}$ and consequently bounded (see Theorem XIII.89 in \cite{reedsimon4}).
This fact and the inequality (\ref{aproxinverse}) ensure that, for each ${\tilde n}_0 \in \mathbb N$, there exists
$K_{\tilde{n}_0}> 0$, so that,
\[|E_n(\varepsilon, \theta) | \leq K_{\tilde{n}_0}, \quad \forall \theta \in {\cal C},\]
for all $\varepsilon > 0$ small enough.

Finally, for each $n_0 \in \mathbb N$, there exists $K_{n_0} > 0$ so that
\[\left|E_n(\varepsilon, \theta) - \nu_n(\theta)\right| \leq 
K_{n_0} \, \varepsilon, \quad n=1, 2 \cdots, n_0, \forall \theta \in {\cal C},\]
for all $\varepsilon > 0$ small enough.

\section{Existence of band gaps; proof of Theorem \ref{corgap}}\label{sectionbandgaps}

This section is dedicated to the proof of Theorem \ref{corgap}.
The steps are similar to those in \cite{yoshitomi}.
In that work, the author studied the
band gap of the spectrum of the Dirichlet Laplacian in a planar periodically curved strip.

Consider the operator 
\[Tw = - w'' + \frac{h''(s)}{h(s)} w + c w, \quad \dom T = H^2(\mathbb R).\]

Recall we have denoted by $\nu_n(\theta)$ th $n$th eigenvalue of $T^\theta$.
By Theorem XIII.89 in \cite{reedsimon4}, 
each $\nu_n(\theta)$ is a continuous function in ${\cal C}$. Furthermore,

\vspace{0.3cm}
\noindent
(a) $\nu_n(\theta) = \nu_n(-\theta)$, for all $\theta \in {\cal C}$, $n=1,2,3, \cdots$.

\vspace{0.3cm}
\noindent
(b) For $n$ odd (resp. even), $\nu_n(\theta)$ is strictly  monotone increasing (resp. decreasing) as $\theta$
increases from $0$ to $\pi/L$. In particular,
\[\nu_1(0) < \nu_1(\pi/L) \leq \nu_2(\pi/L) < \nu_2(0) \leq \cdots
\leq \nu_{2n-1}(0) < \nu_{2n-1}(\pi/L)  \]
\[ \leq \nu_{2n}(\pi/L) < \nu_{2n}(0) \leq \cdots.\]

Now, for each $n=1,2,3,\cdots$, define
\[ B_n := \left\{
\begin{array}{cc}
\left[ \nu_n(0), \nu_n(\pi/L) \right],  & \hbox{for} \, \, \,  n \,\,\, \hbox{odd}, \\
\left[ \nu_n(\pi/L), \nu_n(0) \right],  & \hbox{for} \, \, \, n \,\, \, \hbox{even},
\end{array}\right.
\]
and
\[ G_n := \left\{
\begin{array}{l}
\left(\nu_n(\pi/L), \nu_{n+1}(\pi/L) \right),  \,\,\,
\hbox{for} \, \, \,  n \,\,\, \hbox{odd so that} \,\,\, \nu_n(\pi/L) \neq \nu_{n+1}(\pi/L), \\
\left( \nu_n(0), \nu_{n+1}(0) \right),  \,\,\,
\hbox{for} \, \, \, n \,\, \, \hbox{even so that} \,\,\, \nu_n(0) \neq \nu_{n+1}(0), \\
\emptyset,  \,\,\,  \hbox{otherwise}. 
\end{array}
\right.
\]

By Theorem XIII.90 in \cite{reedsimon4},
we have $\sigma(T) = \cup_{n=1}^{\infty} B_n$;
$B_n$ is called the $j$th band of $\sigma(T)$, and
$G_n$ the gap of $\sigma(T)$ if $B_n\neq \emptyset$.

Corollary \ref{asymptoticbehaviour} implies that for any $n_0 \in \mathbb N$, there exists $\varepsilon_{n_0} > 0$ so that,
for all $\varepsilon \in (0, \varepsilon_{n_0})$,
\[\max_{\theta \in {\cal C}} E_n (\varepsilon, \theta) =
\left\{
\begin{array}{l}
\nu_n(\pi/L) + O(\varepsilon), \,\,\, \hbox{for} \,\,\, n \,\,\, \hbox{odd}, \\
\nu_n(0) + O(\varepsilon), \,\,\, \hbox{for} \,\,\, n \,\,\, \hbox{even},
\end{array}
\right.
\]
and
\[\min_{\theta \in {\cal C}} E_n (\varepsilon, \theta) =
\left\{
\begin{array}{l}
\nu_n(0) + O(\varepsilon), \,\,\, \hbox{for} \,\,\, n \,\,\, \hbox{odd}, \\
\nu_n(\pi/L) + O(\varepsilon),\,\,\, \hbox{for} \,\,\, n \,\,\, \hbox{even},
\end{array}
\right.
\]
hold for each $n=1,2,\cdots,n_0$.
Thus, we have

\begin{Corollary}\label{corcorrected}
For any $n_2 \in \mathbb N$, there exists $\varepsilon_{n_2+1} > 0$ so that,
for all $\varepsilon \in (0, \varepsilon_{n_2+1})$, 
\[\min_{\theta \in {\cal C}} E_{n+1} (\varepsilon, \theta) 
-
\max_{\theta \in {\cal C}} E_n (\varepsilon, \theta) =
|G_n| + O(\varepsilon),\]
holds for $n=1,2, \cdots, n_2$,
where $| \cdot |$ is the Lebesgue measure.
\end{Corollary}

Besides Corollary \ref{corcorrected},
another important point to prove Theorem \ref{corgap} is 
the following result due to Borg \cite{borg}.

\begin{Theorem} (Borg)
Suppose that $W$ is a real-valued, piecewise continuous
function on $[0,L]$. Let $\lambda_n^{\pm}$ be the $n$th eigenvalue of the following operator
counted with multiplicity respectively
\[-\frac{d^2}{ds^2} + W(s), \quad \hbox{in} \quad 
L^2(0,L),\]
with domain
\begin{equation}\label{domainborg}
\{w \in H^2(0,L); w(0)=\pm w(L),
w'(0)=\pm w'(L)\}.
\end{equation}

We suppose that
\[\lambda_n^+ = \lambda_{n+1}^+, \quad
\hbox{for all even} \,\, n,\]
and
\[\lambda_n^- = \lambda_{n+1}^-, \quad
\hbox{for all odd} \,\, n.\]
Then,
$W$ is constant on $[0, L]$.
\end{Theorem}

\noindent
{\bf Proof of Theorem \ref{corgap}:}
For each $\theta \in {\cal C}$, we define the unitary transformation 
$(u_\theta w)(s) = e^{-i\theta s} w(s)$. In particular, consider the
operators  $\tilde{T}^{0}:= u_{0} T^{0} u_{0}^{-1}$
and $\tilde{T}^{\pi/L} := u_{\pi/L} T^{\pi/L} u_{\pi/L}^{-1}$
whose eigenvalues are given by 
$\{\nu_n(0)\}_{n\in \mathbb N}$ and $\{\nu_n(\pi/L)\}_{n\in \mathbb N}$, respectively. Furthermore, the domains of these operators
are given by (\ref{domainborg});
$\tilde{T}^0$ (resp. $\tilde{T}^{\pi/L}$) is called  operator with periodic (resp. antiperiodic) boundary conditions.

Since $h''(s)/h(s)$ is not constant in $[0,L]$, by Borg's Theorem, without loss of generality,
we can say that there exists $n_1 \in \mathbb N$ so that
$\nu_{n_1}(0) \neq \nu_{n_1+1}(0)$.
Now, the result follows by  Corollary  \ref{corcorrected}.


\

\appendix
\section{Appendix}\label{appendixcore}

Let $\mathcal{J}$ be a Hilbert space and 
$b: \dom b \times \dom b \to \mathbb C$ a sesquilinear form  in $\mathcal{J}$.
Denote by $b(\psi)=b(\psi, \psi)$ the quadratic form associated with it.
We say that  $b(\psi)$ is  lower bounded if 
there is $\beta\in \mathbb R$ with 
$b(\psi) \geq \beta \|\psi\|^2$, for all $\psi \in \dom b$.
If $\beta>0$, $b$ is called  positive.
A sesquilinear form $b$ is called hermitian if
$b(\psi, \eta) = b(\eta, \psi)$, for all $\psi, \eta \in \dom b$.

Let $b$ be a hermitian form and $(\psi_n) \subset \dom b$. 
Even though b is not necessarily
positive, this sequence is called a Cauchy sequence with respect to $b$ 
(or in $(\dom b, b)$) if $b(\psi_n - \psi_m) \to 0$ as $n,m \to \infty$. 
It is said that $(\psi_n)$ converges to $\psi$ with
respect to $b$ (or in $(\dom b, b)$) if $\psi \in \dom b$ and 
$b(\psi_n - \psi) \to 0$ as $n \to \infty$.

A sesquilinear form $b$ is closed if for each Cauchy sequence $(\psi_n)$
in $(\dom b, b)$ with $\psi_n \to \psi$ in ${\cal J}$, one has $\psi \in \dom b$ and 
$\psi_n \to \psi$ in $(\dom b, b)$.

Given a sesquilinear form $b$, the operator $T_b$ is associated with $b$ is defined as 
\begin{eqnarray*}
\dom T_b & := & \{ \psi \in  \dom  b: \exists \zeta \in {\cal J} \text{ with } 
b(\eta,\psi)= \langle  \eta, \zeta \rangle, \forall \eta\in \dom b\},\\
T_b \psi & := & \zeta, \quad \psi \in \dom T_b.
\end{eqnarray*}
Thus, $b(\eta,\psi) = \langle \eta,T_b \psi \rangle$, for all $\eta \in \dom b$, 
for all $\psi \in \dom T_b$. 
Such operator is well defined when $\dom b$ is dense in ${\cal J}$.

Recall the quadratic form $t_\varepsilon^\theta(\psi)$ and the operator $T_\varepsilon^\theta$
defined in Section \ref{floquetbloch}.
The goal is to justify that $T_\varepsilon^\theta$ is the self-adjoint operator associated with
$t_\varepsilon^\theta(\psi)$.
The proof is separated in two steps. At first, we prove that
$t_\varepsilon^\theta(\psi)$ is a closed quadratic form. Thus, 
by Theorem 4.2.6 in \cite{cesar}, there exists a self-adjoint operator, denoted by
$T_{t_\varepsilon^\theta}$, so that,
\[t_\varepsilon^\theta(\eta, \psi) = 
\langle \eta, T_{t_\varepsilon^\theta} \psi \rangle, \quad \forall \eta \in \dom t_\varepsilon^\theta, 
\forall \psi \in \dom T_{t_\varepsilon^\theta}.\]
Second, we show that $T_{t_\varepsilon^\theta} = T_\varepsilon^\theta$.

\begin{Proposition}
For each $\theta \in {\cal C}$, the quadratic form $t_\varepsilon^\theta(\psi)$ is closed.
\end{Proposition}
\begin{proof}
We are going to consider the particular case where $\theta=0$ and $k(s)=0$, i.e., $\beta_\varepsilon(s,y)=1$. The general case is similar.

Let $(\psi_n)$ be a Cauchy sequence in $(\dom t_\varepsilon^0, t_\varepsilon^0)$
with $\psi_n \to \psi$ in $L^2(Q, h^2 \ds \dy)$.
In particularly,  since $h$ is a bounded function,
$(\psi_n)$ is a Cauchy sequence in $L^2(Q)$.
We also note that
\[\int_Q |\nabla_y(\psi_n-\psi_m)|^2 \ds\dy \leq \varepsilon^2 \,
t_\varepsilon^0(\psi_n-\psi_m),\]
and 
\begin{eqnarray*}
\int_Q 
|\partial_s(\psi_n-\psi_m)|^2 \ds \dy 
& \leq & 
\frac{1}{(\inf h(s))^2}  \int_Q  h^2|\partial_s(\psi_n-\psi_m)|^2 \ds\dy \\
& \leq & 
\frac{2}{(\inf h(s))^2}
\int_Q h^2 \left|\partialoperator (\psi_n-\psi_m)\right|^2 \ds\dy  \\
&  + & 
2 \int_Q \left|\langle \nabla_y(\psi_n-\psi_m), R^h \rangle\right|^2 \ds \dy \\
& \leq  &
K \left(t^0_\varepsilon(\psi_n,\psi_m) + 
\int_Q \left(|\nabla_y(\psi_n-\psi_m)|^2 + |\psi_n-\psi_m|^2\right) \ds\dy\right),
\end{eqnarray*}
for some $K> 0$.

With theses inequalities, we can see  that $(\psi_n)$ is a Cauchy sequence in
the Hilbert space
${\cal H}^1(Q)$.
Thus,
there exists $\eta \in \mathcal{H}^1(Q)$, so that, $\psi_n \rightarrow \eta$ in $\mathcal{H}^1(Q)$.
We conclude that
$\eta=\psi$ in $L^2(Q)$. Furthermore,
$\partial_s\psi_n\rightarrow\partial_s\psi$,
$\nabla_y \psi_n \to \nabla_y \psi$ in $L^2({Q})$.
	
Now, we are going to show that $\psi(0,y)=\psi(L,y)$ in $L^2(S)$.
Define 
\[V_n(y):=\int^L_0\partial_s\psi_n(s,y)\ds, \quad V(y):=\int^L_0\partial_s\psi(s,y) \ds,\]
and note that 
\begin{eqnarray*}
\int_s|V_n(y)-V(y)| \dy  & \leq  &   \int_Q|\partial_s\psi_n-\partial_s\psi| \ds\dy \\
& \leq  &
 |Q|^{1/2}\left(\int_Q|\partial_s\psi_n-\partial_s\psi|^2 \ds\dy\right)^{1/2}\rightarrow 0,
 \quad n \to \infty.
\end{eqnarray*}

Thus, $V_n\rightarrow V$ in ${\rm L}^1(S)$.
Therefore,
there exists a subsequence $(V_{n_k})$ of $(V_n)$, so that,
$V_{n_k}(y) \rightarrow V(y)$, a.e. $y$. More exactly,
\[\lim_{k\rightarrow\infty}\int^L_0\partial\psi_{n_k}(s,y) \ds=\int^L_0\partial_s\psi(s,y)\ds,
\quad \hbox{a.e. } y.\]

Recall $\psi_{n_k}(L,y)=\psi_{n_k}(0,y)$.
By Fundamental Theorem of Calculus
\[0 = \lim_{k\rightarrow\infty}(\psi_{n_k}(L,y)-\psi_{n_k}(0,y))=\psi(L,y)-\psi(0,y),\quad \text{a.e. } y.\]
Thus, $\psi \in \dom t_\varepsilon^0$.

Finally, we can see that there exists $K>0$, so that, 
\[t^0_\varepsilon(\psi_n-\psi)\leq K\|\psi_n-\psi\|^2_{\mathcal{H}^1(Q)}\rightarrow 0,
\quad n \to \infty,\]
i.e.,
$\psi_n \to \psi$ in $(\dom t_\varepsilon^0, t_\varepsilon^0)$.

\end{proof}

\begin{Proposition}
For each $\theta \in {\cal C}$, $T_\varepsilon^\theta = T_{t_\varepsilon}^\theta$.
\end{Proposition}
\begin{proof}
Again, consider the particular case $\theta=0$ and $k(s)=0$.
Write $R^h=(R^h_1,R^h_2)$, denote 
by
$N=(N_1, N_2)$
the outward pointing unit normal to $S$
and
${\rm d}A$ the measure of area of the region
$\partial S$.

By identity polarization we obtain the sesquilinear form $t_\varepsilon^0(\eta,\psi)$
associated with the quadratic 
form $t_\varepsilon^0(\psi)$. Namely,
\begin{eqnarray*}
t_\varepsilon^0(\eta,\psi) &= & \int_Q 
\left(h^2\partialoperator\overline{\eta}\,\partialoperator\psi+\frac{1}{\varepsilon^2}\langle\nabla_y\overline{\eta},\nabla_y\psi\rangle\right) \ds\dy \\
& = & \int_Q h^2\partial_s\overline{\eta\,}\partialoperator\psi \ds\dy +
\int_Q h^2\langle\nabla_y\overline{\eta},R^h\rangle\partialoperator\psi \ds\dy \\ 
& + & \int_Q \frac{1}{\varepsilon^2}\langle\nabla_y\overline{\eta},\nabla_y\psi\rangle \ds\dy
+ c \int_Q h^2 \overline{\eta}  \psi \,\ds \dy.
\end{eqnarray*}

For each $\eta \in \dom t_\varepsilon^0$ 
and $\psi \in \dom t_\varepsilon^0 \cap H^2(Q)$, the Fubini Theorem and an integration by parts 
show that
\begin{eqnarray*}
& & 
\int_Q h^2\partial_s \overline{\eta} \, \partialoperator\psi \ds\dy 
=  
-\int_Q \overline{\eta} \, \partial_s\left(h^2\partialoperator\psi\right) \ds\dy
+
\int_S \left( \overline{\eta} \, h^2 \partialoperator\psi \right) |_0^L  \, \dy = \\
& -  &
\int_Q \overline{\eta} \, \partial_s\left(h^2\partialoperator\psi\right) \ds\dy
+
\int_S \overline{\eta}(0,y) h^2(0) 
\left( \partial_{s,y}^{Rh} \psi(L,y) -  \partial_{s,y}^{Rh} \psi(0,y) \right)   \, \dy. 
\end{eqnarray*}
Furthermore,
\begin{eqnarray*}
& &
\int_Q h^2\langle\nabla_y\overline{\eta},R^h\rangle\partialoperator\psi \ds\dy =  \\
&  & 
\int_Q (\partial_{y_1}\overline{\eta})R^h_1h^2\partialoperator\psi \ds\dy	
+\int_Q (\partial_{y_2}\overline{\eta})R^h_2h^2\partialoperator\psi \ds\dy  = \\
&  &
-\int_Q \overline{\eta}\,\partial_{y_1}\left(R^h_1h^2\partialoperator\psi\right) \ds\dy
+\int_0^L \int_{\partial S} \overline{\eta}\,R^h_1h^2\partialoperator\psi N_1\, dA \ds \\
&  &
-
\int_Q \overline{\eta}\partial_{y_2}\left(R^h_2h^2\partialoperator\psi\right) \ds\dy
+\int_0^L \int_{\partial S} \overline{\eta}\,R^h_2h^2\partialoperator\psi N_2\, dA \ds =\\
&  &
-\int_Q \overline{\eta}\,{\rm div}_y\left(R^h h^2\partialoperator\psi\right) \ds\dy
+\int_0^L \int_{\partial S} \overline{\eta}\langle R^h,N\rangle h^2\partialoperator\psi\, dA \ds,
\end{eqnarray*}
and 
\[ \int_Q\frac{1}{\varepsilon^2}\langle\nabla_y\eta,\nabla_y\psi\rangle\, \ds\dy 
=
-\int_Q\frac{1}{\varepsilon^2}\overline{\eta} \,  \Delta_y\psi \, \ds\dy 
+\int_0^L \int_{\partial S} \frac{1}{\varepsilon^2}\overline{\eta}\langle\nabla_y\psi,N\rangle\, dA \ds.\]

Thus, 
\begin{eqnarray*}
t_\varepsilon^0(\eta,\psi)\quad & = &
-\int_Q\overline{\eta}\left[\left(\partial_s+{\rm div}_yR^h\right)h^2\partialoperator\psi+\frac{1}{\varepsilon^2}\Delta_y\psi\right] \ds \dy  \\
& + &
\int_S \overline{\eta}(0,y) h^2(0) 
\left( \partial_{s,y}^{Rh} \psi(L,y) -  \partial_{s,y}^{Rh} \psi(0,y) \right)   \, \dy \\
& + &
\int_0^L \int_{\partial S}
\overline{\eta} \left( h^2 \langle R^h,N \rangle \partialoperator \psi  
+ \frac{1}{\varepsilon^2} \langle\nabla_y\psi, N\rangle \right)\,dA \ds
+ c \int_Q h^2 \overline{\eta} \psi \ds \dy.
\end{eqnarray*}

For $\psi \in \dom t_\varepsilon^0 \cap H^2(Q)$, we define
\[Z_\varepsilon^0 \psi := -\frac{1}{h^2}
\left[\left(\partial_s+
{\rm div}_yR^h\right)h^2\partialoperator\psi+\frac{1}{\varepsilon^2}\Delta_y\psi\right] + c \psi.\]
Therefore,
\begin{eqnarray} \label{defzopassquf}
t_\varepsilon^0(\eta,\psi)
& = &
\langle\eta, Z_\varepsilon^0 \psi\rangle_{\mathcal{H}}  +
\int_S \overline{\eta}(0,y) h^2(0) 
\left( \partial_{s,y}^{Rh} \psi(L,y) -  \partial_{s,y}^{Rh} \psi(0,y) \right)   \, \dy \nonumber \\
& + &
\int_0^L \int_{\partial S}
\overline{\eta} \frac{\partial^{Rh} \psi}{\partial N} \,dA \ds,
\end{eqnarray}   
for all $\eta\in \dom t_\varepsilon^0$, for all $\psi \in \dom t_\varepsilon^0 \cap H^2(Q)$.

\vspace{0.3cm}
\noindent
{\bf Step 1:}
Given $\psi \in \dom T_\varepsilon^0$, we have
$(\partial^{Rh} \psi/\partial N) = 0$ on $[0,L) \times \partial S$ and,
\[t_\varepsilon^0(\eta, \psi) = \langle \eta, T_\varepsilon^\theta  \psi \rangle_{{\cal H}_\varepsilon'}, \quad
\forall \eta \in \dom t_\varepsilon^0.\]
Thus, $ \psi\in \dom T_{t_\varepsilon^0}$ and $T_{t_\varepsilon^0} \psi = T_\varepsilon^0 \psi$. 

\vspace{0.3cm}
\noindent
{\bf Step 2:}
Conversely, take $\psi \in \dom T_{t_\varepsilon^0} \subset \dom t_\varepsilon^0$. Then,
there exists $\zeta \in {\cal H}$, so that,
\[	t_\varepsilon^0(\eta,\psi)=\langle \eta, \zeta \rangle_{{\cal H}_\varepsilon'},\quad \forall \eta \in \dom t_\varepsilon^0.\]
This implies that $\psi \in H^2(Q)$ (see Chapter 7 in \cite{baiocchi}) and, by (\ref{defzopassquf}),
\[\langle \eta, \zeta - Z_\varepsilon^0 \psi \rangle_{{\cal H}_\varepsilon'} =
\int_S \overline{\eta}(0,y) h^2(0) 
\left( \partial_{s,y}^{Rh} \psi(L,y) -  \partial_{s,y}^{Rh} \psi(0,y) \right)   \, \dy
+
\int_0^L \int_{\partial S}
\overline{\eta} \frac{\partial^{Rh}\psi}{\partial N}  \,dA \ds.\]

In particularly, 
\[\langle \eta, \zeta - Z_\varepsilon^0 \psi \rangle_{{\cal H}_\varepsilon'} = 0, \quad
\forall \eta \in C_0^\infty(Q) \subset \dom t_\varepsilon^0.\]

Therefore, $\zeta = Z_\varepsilon^0 \psi$.
It remains to show that $\psi \in \dom T_\varepsilon^0$.

We know that $\psi(0,y)=\psi(L,y)$ in $L^2(S)$. On the other hand,
since $\zeta = Z_\varepsilon^0 \psi$, 
\[\int_S \overline{\eta}(0,y) h^2(0) 
\left( \partial_{s,y}^{Rh} \psi(L,y) -  \partial_{s,y}^{Rh} \psi(0,y) \right)   \, \dy
+
\int_0^L \int_{\partial S}
\overline{\eta} \frac{\partial^{Rh} \psi}{\partial N}\,dA \ds = 0,\]
for all $\eta \in \dom t_\varepsilon^0$.
By taking $\eta(s,y) = w(s) u(y)$, with $w \in C_0^\infty(0,L)$ and $u \in H^1(S)$,
\[\int_0^L  w(s) \int_{\partial S} u(y)
\frac{\partial^{Rh} \psi}{\partial N}  \,dA \ds = 0,
\quad \forall w \in C_0^\infty(0,L), \forall u \in H^1(S).\]
Thus,
\begin{equation}\label{appconin01}
\frac{\partial^{Rh} \psi }{\partial N} = 0, \quad
\hbox{in} \quad L^2(Q).
\end{equation}

Consequently,
\[\int_S \overline{\eta}(0,y) h^2(0) 
\left( \partial_{s,y}^{Rh} \psi(L,y) -  \partial_{s,y}^{Rh} \psi(0,y) \right)   \, \dy = 0,
\quad \forall \eta \in \dom t_\varepsilon^0.\]
With suitable choices of $\eta$, one can show
\begin{equation}\label{appconin02}
\partial_{s,y}^{Rh} \psi(L,y) =  \partial_{s,y}^{Rh} \psi(0,y), \quad
\hbox{in} \quad L^2([0,L)\times \partial S).
\end{equation}

The fact that $\psi(0,y)=\psi(L,y)$ in $L^2(S)$, together
with the conditions (\ref{appconin01}) and (\ref{appconin02}),
ensures that 
$\psi \in \dom T_\varepsilon^0$.
\end{proof}

\begin{Remark}{\rm
Recall the quadratic form $t_\varepsilon(\psi)$ and the operator 
$T_\varepsilon$ defined in Section \ref{coordinates}.
Similarly, one can show that $t_\varepsilon(\psi)$ is a closed quadratic form
and $T_\varepsilon$ is the self-adjoint operator associated with it.
The proof will be omitted in this text.}
\end{Remark}


%
%


\begin{thebibliography}{00}


\bibitem{baiocchi}
Baiocchi, C. and Capelo, A.:
Variational and quasivariational inequalities: applications to free boundary problems.
Published by John Wiley $\&$ Sons Ltd (1984)

\bibitem{bdv} Bedoya, R., de Oliveira, C. R. and Verri A. A.: Complex $\Gamma-$convergence and magnetic Dirichlet Laplacian in bounded thin tubes. J. Spectr. Theory \textbf{4} (2014), 621-642. 

\bibitem{bde} Bentosela, F., Duclos, P. and Exner, P.: Absolute continuity in periodic thin tubes and strongly coupled leaky wires. Lett. in Math. Phys. \textbf{65} (2003), 75-82.




\bibitem{borg}
Borg, G.: Eine Umkehrung der Sturm--Liouvillschen Eigenwertaufgabe. Bestimmung der
Differentialgleichung durch die Eigenwerte, Acta Math. {\bf 78} (1946), 1--96.







\bibitem{cesar} de Oliveira. C. R.: Intermediate Spectral Theory and Quantum Dynamics. Birkh\"{a}user, Basel (2009)






\bibitem{dov4} de Oliveira, C. R. and Verri, A. A.: Asymptotic spectrum for Dirichlet Laplacian in thin Deformed tubes with scaled geometry. 
 J. Phys. A: Math. and Theor.,
{\bf 45} (2012)  p.435201.



\bibitem{friedabscon} Friedlander, L.: Absolute continuity of the spectra of periodic waveguides. Contemporary Mathematics, \textbf{339} (2003), 37-42.

\bibitem{solomyak} Friedlander, L. and Solomyak, M.:  
On the spectrum of the Dirichlet Laplacian in a narrow infinite strip.
Amer. Math. Soc. Transl. {\bf 225}  (2008), 103--116.

\bibitem{friedlandersolomyak2} Friedlander, L. and Solomyak, M.: On the spectrum of narrow periodic waveguide. Russ. J. Math. Phys. \textbf{15} (2008), 238-242. 




\bibitem{gohberg}
Gohberg, I. C. and Krein, M. G.: Introduction to the theory of linear nonselfadjoint
operators. Translations of Mathematical Monographs 18, American Mathematical
Society (1969)



\bibitem{kato} Kato, T.: Perturbation Theory for linear Operators. Springer-Verlag, Berlin (1995)



\bibitem{klingenberg} Klingenberg , W.: A Course in Differential Geometry. Springer-Verlag, New  York (1978) 


\bibitem{nazarov}
Nazarov, S. A.: A Gap in the Essential Spectrum of the Neumann Problem
for an Elliptic System in a Periodic Domaim.
Functional Analysis and Its Applications, {\bf 43}, No. 3, (2009), 239--241.
Translated from Funktsional'nyi Analiz i Ego Prilozheniya, {\bf 43}, No. 3, (2009), 92--95.

\bibitem{reedsimon4} Reed, M, and Simon, B.: Methods of Modern Mathematical Physics, IV. Analysis of Operator. Academic Press, New York (1978)


\bibitem{sobolev} Sobolev, A. V. and Walthoe, J.: Absolute continuity in periodic waveguides. Proc. London Math. Soc. \textbf{85} (2002), 717-741.

\bibitem{vema}
Verri, A. A. and Mamani, C. R.: 
Absolute continuity in periodically bent and twisted tubes.
http://arxiv.org/abs/1508.02574



\bibitem{yoshitomi} Yoshitomi, K.: Band gap of the spectrum in periodically curved quantum waveguides. J. Differ. Equations \textbf{142} (1998), 123-166.
















\end{thebibliography}
\end{document}